\documentclass{article}
\usepackage[english]{babel}
\usepackage[utf8x]{inputenc}
\usepackage[T1]{fontenc}
\usepackage{algorithm}
\usepackage{algpseudocode}
\usepackage{amsmath}
\usepackage{amssymb}
\usepackage{amsthm}
\usepackage{mathtools}
\newtheorem{theorem}{Theorem}
\newtheorem{lemma}[theorem]{Lemma}

\title{On the complexity of solving Subtraction games}
\author{Kamil Khadiev\thanks{
Center for Quantum Computer Science, Faculty of Computing, University of Latvia, Riga, Latvia; Kazan Federal University, Kazan, Russia; e-mail: kamilhadi@gmail.com}
\and 
Dmitry Kravchenko\thanks{Center for Quantum Computer Science, Faculty of Computing, University of Latvia, Riga, Latvia; e-mail: kravchenko@gmail.com}}

\begin{document}
\maketitle

\begin{abstract}
We study algorithms for solving Subtraction games, which sometimes are referred to as one-heap Nim games.

We describe a quantum algorithm which is applicable to any game on DAG, and show that its query compexity for solving an arbitrary Subtraction game of $n$ stones is $O\left(n^{3/2}\log n\right)$.

The best known deterministic algorithms for solving such games are based on the dynamic programming approach \cite{cormen2001}. We show that this approach is asymptotically optimal and that classical query complexity for solving a Subtraction game is generally $\Theta\left(n^2\right)$.

This paper perhaps is the first explicit ``quantum'' contribution to algorithmic game theory.
\end{abstract}

\section{Introduction}

A Subtraction game is similar to a canonical Nim game \cite{f2000}. The difference is that players deal with just one heap of stones but with some limitations imposed on the number of stones they can take from the heap. The most common limitation is defining maximum for the number of stones to be taken away, and this kind of games has nice combinatorial solutions.
Here we study a much more general class of such limitations and thus a broader class of Subtraction games.

\section{Definitions}

\subsection{Subtraction games}
Let us formally define what we call a Subtraction game throughout this paper:
It is a two-player game in which the players alternately remove some positive amounts of stones from a heap. Let $n$ be the initial number of stones in the heap, and $\Gamma$ be a (triangular) binary matrix of size $n\times n$. A player can remove $j-i$ stones once there remain exactly $j$ stones in the heap ($0 \le i < j \le n$) iff $\Gamma_{j,i}=1$. One who cannot make a legal move loses, and their opponent wins the game. Thus a player wins a game if he takes all the remaining stones, or leaves such number of stones $j$ that no allowed moves remain: $\sum_{i=0}^{j-1} \Gamma_{j,i} = 0$.

Obviously, the rules of a Subtraction game are fully determined by such matrix $\Gamma$, so hereafter we sometimes refer to $\Gamma$ as to a game. We also reserve name $n$ to denote the initial size of stones, which also corresponds to the dimension of the matrix $\Gamma$. Note that there are only $n\left(n+1\right)/2$ meaningful bits in the matrix $\Gamma$, as a player cannot take more stones than there remain in the heap: $\Gamma_{j,i}=0$ for all $j<i$.

\subsection{Winning function}
We define a Boolean function $\textsc{Win}$: ${\textsc{Win}\left(\Gamma\right)=1}$ iff the first player has a winning strategy in game $\Gamma$.

We also extend its domain to all possible positions in the game: ${\textsc{Win}\left(\Gamma,j\right)=1}$ iff a player having $j$ stones in game $\Gamma$ has a winning strategy. In particular:
\begin{itemize}
\item $\textsc{Win}\left(\Gamma,n\right)=\textsc{Win}\left(\Gamma\right)$; \item $\textsc{Win}\left(\Gamma,0\right)=0$;
\item $\textsc{Win}\left(\Gamma,j\right) \iff \exists i:\Gamma_{j,i} \land \neg \textsc{Win}\left(\Gamma,i\right)$.
\end{itemize}

Hereinafter we alternately use Boolean and integer forms of these and other values: $\mathrm{TRUE}=1$ and $\mathrm{FALSE}=0$.

\subsection{Properties of Subtraction games}

We call a game $\Gamma$ \emph{$k$-balanced} if the number of winning positions differs from the number of losing positions by at most $k$:
\begin{equation}\label{eq:balanced}
\Big|n/2-\sum_{j=1}^n{\textsc{Win}\left(\Gamma,j\right)}\Big| \le k/2
\end{equation}

We call a game $\Gamma$ \emph{sensitive} if in each winning position a player has a unique winning move. Or, equivalently, if in each position a player can have at most one winning move:
\begin{equation}\label{eq:sensitive}
\forall j: \Big|\left\{i:\Gamma_{j,i} \land \neg \textsc{Win}\left(\Gamma,i\right)\right\}\Big| \leq 1
\end{equation}

In the next section we study sensitive $k$-balanced games for small $k$-s.
Hereafter one can assume $k=0$ for simplicity, but all our considerations also hold for any $k=o\left(n\right)$. The exact value of $k$ only affects the size of the considered subset of games.

\section{Classical query compexity}

\begin{lemma}\label{sensitivitylemma}
Let $\Gamma$ be a losing sensitive $o\left(n\right)$-balanced Subtraction game $\Gamma$ picked uniformly at random: $\eqref{eq:balanced} \land \eqref{eq:sensitive} \land \neg \textsc{Win}\left(\Gamma\right)$. Let game $\Gamma'$ differ from $\Gamma$ in exactly one random bit of the binary representation: $\textsc{HammingDistance}\left(\Gamma,\Gamma'\right)=1$. Then $\mathbb{E}\big[\textsc{Win}\left(\Gamma'\right)\big] \ge 1/24.$
\end{lemma}
\begin{proof}

We first make three assumptions for $\forall j,i \left(0 \le i < j \le n\right)$:
\begin{itemize}
\item $\mathbb{E}\big[\textsc{Win}\left(\Gamma,j\right)\big]=1/2$, which does not fully correspond to the uniform distribution for the considered subset of games, but is asymptotically equivalent to it (i.e. expected difference between the two pmf-s is neglectable for large $n$-s). We leave this fact without formal proof, because even non-uniform distribution is sufficient for our purposes. We also neglect addend $\pm k/n=\pm o\left(1\right)$ in this estimation.
\item $\textsc{Win}\left(\Gamma,i\right) \implies \mathbb{E}\big[\Gamma_{j,i}]=1/2$. Informally: a winning position $i$ is achievable from any preceding position $j$ with probability $1/2$. This assumption is perfectly correct since possibility or impossibility to make a losing move in any position leaves a game strategically unaffected.
\item $\textsc{Win}\left(\Gamma,j\right) \land \neg \textsc{Win}\left(\Gamma,i\right) \implies \mathbb{E}\big[\Gamma_{j,i}]=1/\sum_{i'=0}^{j-1}\neg \textsc{Win}\left(\Gamma,i'\right)$. Informally: from a winning position $j$, all subsequent losing positions $i'$ are achievable equiprobably (and these probabilities sum up to $1$). This assumption also is perfectly correct due to the definition of the considered subset of games.
\end{itemize}

Now let $\Gamma'_{j,i} \ne \Gamma_{j,i}$ for some pair of indices $j,i$ picked at random. Then:
\begin{itemize}
  \item $\mathrm{Pr}\big[\neg \textsc{Win}\left(\Gamma,j\right)\big] = 0.5$
  \item $\forall j' (j<j'<n):$
    \begin{sloppypar}$\mathrm{Pr}\big[\textsc{Win}\left(\Gamma,j'\right) \land \Gamma_{j',j} \mid \neg \textsc{Win}\left(\Gamma,j\right)\big] = 0.5/\mathbb{E}\big[\sum_{i'=0}^{j'-1}\neg \textsc{Win}\left(\Gamma,i'\right)\big] = \displaystyle\frac{1}{j'}$\end{sloppypar}
    \begin{sloppypar}Informally: each preceding $j'$ gives some small chance of $1/j'$ for a losing position $j$ to be accessible from a winning position $j'$.\end{sloppypar}
  \item $\mathrm{Pr}\big[\nexists j':\textsc{Win}\left(\Gamma,j'\right) \land \Gamma_{j',j} \mid \neg \textsc{Win}\left(\Gamma,j\right)\big] = \prod_{j'=j+1}^{n-1} \left(1-\displaystyle\frac{1}{j'}\right)=\displaystyle\frac{j}{n-1}$
    \begin{sloppypar}Informally: even though each previously estimated chance was small, alltogether they result in some significant probability for a losing position $j$ to be accessible from at least one preceding winning position (of totally $n-j-1$ preceding positions).\end{sloppypar}
\end{itemize}

We follow that once $j$ is a losing position, with probability $1-\frac{j}{n-1}$ it is achievable from some winning position $j'$. We shall consider the least $j'$ if there happen to occur several (this can be important!).

On the other side, $\mathrm{Pr}\big[\neg \textsc{Win}\left(\Gamma,i\right)\big] = 0.5$, so once $j$ is a losing position, then with probability $0.5$ position $i$ also is losing, and thus $\Gamma_{j,i}=0$ (since one losing position cannot be achievable from another losing position). And alltogether: with probability
\begin{equation}\label{eq:prob1}
\left(1-\displaystyle\frac{j}{n-1}\right) \times 0.5 \times 0.5
\end{equation}
we have $\Big(\textsc{Win}\left(\Gamma,j'\right) \land \Gamma_{j',j}\Big) \land \neg \textsc{Win}\left(\Gamma,j\right) \land \neg \textsc{Win}\left(\Gamma,i\right)$.

Now let us observe what happens with $\Gamma'$ in this (not so improbable) case: we have $\textsc{Win}\left(\Gamma',i\right)=\textsc{Win}\left(\Gamma,i\right)=0$ and $\Gamma'_{j,i}=\neg \Gamma_{j,i}=1$. It implies $\textsc{Win}\left(\Gamma',j\right)$ since $\Gamma'$ allows access from $j$ to the losing position $i$. And then the only winning move in position $j'$ becomes obsolete: $\textsc{Win}\left(\Gamma',j'\right)=0$ (some other winning positions of game $\Gamma$ also may become losing in $\Gamma'$, but it definitely should happen with the ``smallest'' former winning position $j'$).

Finally we note that with probability $\mathrm{Pr}\big[\Gamma_{n,j'}=1\big] = 0.5$ this difference also implies $\textsc{Win}\left(\Gamma',n\right) = 1 \neq \textsc{Win}\left(\Gamma,n\right)$. This probability together with \eqref{eq:prob1} results in the overall lower bound:
\begin{align*}
  \mathrm{Pr}\big[\textsc{Win}\left(\Gamma',n\right) \neq \textsc{Win}\left(\Gamma,n\right)\big] & \ge \mathbb{E} \big[ 0.5 \times \left(1-\frac{j}{n-1}\right) \times 0.5 \times 0.5 \big] \\
  & = \left(1 - \frac{\mathbb{E}\left[j\right]}{n-1}\right)/8 \\
  & \approx \frac{1}{24} \\
\end{align*}
\end{proof}

\begin{theorem}
There is no deterministic or randomized algorithm for solving function $\textsc{Win}$ faster than in $O\left(n^2\right)$ steps.
\end{theorem}
\begin{proof}
We shall prove this impossibility by analyzing performance of the best classical algorithm on the set $\textsc{LSB}$ of losing sensitive $o\left(n\right)$-balanced Subtraction games.
But perhaps we first need to make a couple of remarks on the eligibility of such proof:
\begin{itemize}
\item It may seem that $\textsc{LSB}$ is a too small set, and that we are going to prove the lower bound for some negligible number of games. However, actually the number of such games is roughly $2^{n^2/4}$ while the total number of Subtraction games is roughly $2^{n^2/2}$. That is, a game from $\textsc{LSB}$ and an arbitrary Subtraction game are representable by asymptotically similar numbers of bits. Although just one example would be sufficient to make a statement about the ``worst case'' complexity, we here demonstrate that a \textit{significant} part of all Subtraction games are hard to solve. We also note that we strongly believe that similar bound should also hold for the whole set of Subtraction games.
\item It may seem that since $\textsc{LSB}$ contains only losing games ($\textsc{Win}\left(\Gamma\right)=0$), one could design an algorithm which somehow recognizes that a game belongs to this set and returns answer $0$. Actually we can easily refute this criticism by extending this set with sufficient number of $\Gamma'$-s which often are winning and always are hardly-distinguishable from the losing games of $\textsc{LSB}$. Formally, we can consider set $\textsc{LSB}' = \displaystyle\bigcup_{\Gamma \in \textsc{LSB}}\textsc{HammingBall}\left(\Gamma,1\right)$.
\end{itemize}

Suppose a classical algorithm, given $\Gamma \in \textsc{LSB}$, reports an answer after querying on average less than $n\left(n+1\right)/2/24$ bits of $\Gamma$.

Obviously, such algorithm cannot pretend to be correct with probability more than $0.5$, since Lemma~\ref{sensitivitylemma} implies that there are on average $n\left(n+1\right)/2/24$ such crucial bits $\Gamma_{j,i}$ that inverting the bit also inverts the value of the game $\textsc{Win}\left(\Gamma\right)$. Leaving any such bit unchecked means failing to guess $\textsc{Win}\left(\Gamma\right)$ with adequate probability.

\end{proof}

\section{Quantum Algorithm}

In this section we suggest a quantum algorithm for solving an arbitrary Subtraction game (should it belong to $\textsc{LSB}$ or not). We consider directed acycling graph (DAG) $G=\left(V,E\right)$ with adjacency matrix $\Gamma$: set $V$ corresponds to $n+1$ positions of the game $\Gamma$, and set $E$ corresponds to all legal moves.

The algorithm applies dynamic programming approach \cite{k2018,cormen2001}: it solves a problem using precomputed solutions of smaller parts of the same problem. For analyzing DAGs it typically means usage of some modification of Depth-first search algorithm (DFS) as a subroutine \cite{cormen2001}.

We exploit the same idea, but for solving $\textsc{Win}\left(\Gamma,j\right)$ (for each vertex $j$, starting from $j=1$) we use Grover's Search algorithm \cite{g96,bbht98} to find a losing vertex $i'$ among directly accessible vertices $\textsc{Adj}\left[j\right] \stackrel{\text{def}}{=} \left\{i: \Gamma_{j,i}\right\}$, s.t. $\textsc{Win}\left(\Gamma,i'\right)=0$. This algorithm has two important properties:
\begin{itemize}
\item its time complexity is $O\left(\sqrt{\deg{j}}\right)$, where $\deg{j} \stackrel{\text{def}}{=} \big|\textsc{Adj}\left[j\right]\big|$ is the number of vertices directly accessible from the vertex $j$;
\item it returns the desired losing vertex $i'$ with some constant probability (say $0.5$) if such vertex exists.
\end{itemize}
In Algorithm \ref{quantumalgo} we use Grover's Search in form of function $\textsc{Grover\_IsZeroAmong}$ which returns $0$ or $1$ equiprobably if there is zero among its input bits, and returns $0$ if all inputs are ones. We store the search results in array $w$ and reuse them in all the subsequent searches.

\begin{algorithm}
\caption{Quantum Algorithm for solving $\textsc{Win}\left(\Gamma\right)$}\label{quantumalgo}
\begin{algorithmic}
\State $w_0 \gets 0$
\For{$j=1\ldots n$}\Comment{$O\left(n\right)$}
  \State $w_j \gets 0$
  \For{$z=1\ldots 2 \cdot \log_2{n}$}\Comment{$O\left(\log{n}\right)$}
    \If{$\textsc{Grover\_IsZeroAmong} \left\{w_i \mid i \in \textsc{Adj}\left[j\right]\right\}$}\Comment{$O\left(\sqrt{\deg{j}}\right)$}
      \State $w_j \gets 1$
    \EndIf
  \EndFor
\EndFor
\State \Return $w_n$
\end{algorithmic}
\end{algorithm}

\begin{theorem}\label{th:and-or-dag}
Algorithm~\ref{quantumalgo} computes $\textsc{Win}\left(\Gamma\right)$ in time $O\left(\sqrt{n\left|E\right|}\log n\right)$ and with error probability $\epsilon \lesssim 1/n$.
\end{theorem}
\begin{proof}$ $

Correctness of the algorithm is obvious.

Time complexity follows from Cauchy-Bunyakovsky-Schwarz inequality:
$$ \sum_{j=1}^{n}{\sqrt{\deg{j}}} \le \sum_{j=1}^{n}{\sqrt{\mathbb{E}_j\left[\deg{j}\right]}}=\sum_{j=1}^{n}{\sqrt{\left|E\right|/n}}=\sqrt{n\left|E\right|}.$$

\sloppy{The probability of error in evaluating one particular $w_j$ is ${2^{-2\log_2 n}=1/{n^2}}$, so the probability of no error at all among evaluations of $w_1, \ldots, w_n$ is ${\left(1-1/{n^2}\right)^n \gtrsim 1-1/n}$.}

\end{proof}

We note that for a random Subtraction game the expected number of edges $\mathbb{E}\big[\left|E\right|\big] = O\left(n^2\right)$ and formulate the final conclusion: while the best classical algorithms require time $\Theta\left(n^2\right)$ to solve a Subtraction game, there exists a polynomially faster quantum algorithm which runs in time $O\left(n^{3/2}\log{n}\right)$.

\section*{Acknowledgement}
The work is performed according to the Russian Government Program of Competitive Growth of Kazan Federal University.

The research is supported by ERC Advanced Grant MQC, PostDoc Latvia Program, and by the ERDF within the project 1.1.1.2/VIAA/1/16/099 ``Optimal quantum-entangled behavior under unknown circumstances''.

\end{document}